\documentclass[onecolumn,superscriptaddress,notitlepage]{revtex4}

\usepackage{bm}

\usepackage{amsmath}
\usepackage{graphicx}
\usepackage{amsbsy,amsfonts}
\usepackage{amsthm}
\usepackage{natbib}
\usepackage{colonequals}
\usepackage[colorlinks]{hyperref}
\usepackage{hypcap}

\input{Qcircuit.tex}

\newcommand{\nn}{\nonumber \\}

\newcommand{\defeq}{\colonequals}

\newcommand{\ignore}[1]{}

\newtheorem{theorem}{Theorem}
\newtheorem{lemma}[theorem]{Lemma}

\newtheorem{definition}{Definition}
\newtheorem{corollary}[theorem]{Corollary}
\newenvironment{proofof}[1]{\begin{proof}[Proof of~#1.]}{\end{proof}}

\newcommand{\eq}[1]{\hyperref[eq:#1]{(\ref*{eq:#1})}}
\renewcommand{\sec}[1]{\hyperref[sec:#1]{Section~\ref*{sec:#1}}}
\newcommand{\app}[1]{\hyperref[app:#1]{Appendix~\ref*{app:#1}}}
\newcommand{\fig}[1]{\hyperref[fig:#1]{Figure~\ref*{fig:#1}}}
\newcommand{\thm}[1]{\hyperref[thm:#1]{Theorem~\ref*{thm:#1}}}
\newcommand{\lem}[1]{\hyperref[lem:#1]{Lemma~\ref*{lem:#1}}}
\newcommand{\cor}[1]{\hyperref[cor:#1]{Corollary~\ref*{cor:#1}}}
\newcommand{\defn}[1]{\hyperref[def:#1]{Definition~\ref*{def:#1}}}

\newcommand{\R}{\mathbb{R}}
\newcommand{\C}{\mathbb{C}}

\newcommand{\Rem}{\mathop{\mathbf{R}}\nolimits}
\newcommand{\norm}[1]{\|{#1}\|}
\newcommand{\Norm}[1]{\left\|{#1}\right\|}
\DeclareMathOperator{\poly}{poly}

\begin{document}
\title{Hamiltonian Simulation Using Linear Combinations of Unitary Operations}
\author{Andrew M. Childs}
\affiliation{Department of Combinatorics \& Optimization, University of Waterloo, Ontario N2L 3G1, Canada}
\affiliation{Institute for Quantum Computing, University of Waterloo, Ontario N2L 3G1, Canada}
\author{Nathan Wiebe}
\affiliation{Institute for Quantum Computing, University of Waterloo, Ontario N2L 3G1, Canada}

\begin{abstract}
We present a new approach to simulating Hamiltonian dynamics based on implementing linear combinations of unitary operations rather than products of unitary operations.  The resulting algorithm has superior performance to existing simulation algorithms based on product formulas and, most notably, scales better with the simulation error than any known Hamiltonian simulation technique.  Our main tool is a general method to nearly deterministically implement linear combinations of nearby unitary operations, which we show is optimal among a large class of methods.
\end{abstract}
\maketitle


\section{Introduction}
Simulating the time evolution of quantum systems is a major potential application of quantum computers.  While quantum simulation is apparently intractable using classical computers, quantum computers are naturally suited to this task. Even before a fault-tolerant quantum computer is built, quantum simulation techniques can be used to prove equivalence between Hamiltonian-based models of quantum computing (such as adiabatic quantum computing~\cite{ADKLLR07} and continuous-time quantum walks~\cite{Chi09a}) and to develop novel quantum algorithms~\cite{FGGS00,CCDFGS03,FGG07,HHL09,B10}.

In recent years there has been considerable interest in optimizing quantum simulation algorithms.  The original approach to quantum simulation, based on product formulas, was proposed by Lloyd for time-independent local Hamiltonians~\cite{lloyd}.  This work was later generalized to give efficient simulations of sparse time-independent Hamiltonians that need not have tensor-product structure~\cite{CCDFGS03,ATS03}.  Further refinements of these schemes have yielded improved performance~\cite{Chi04,BACS07,PaZh10,CK11} and the techniques have been extended to cover time-dependent Hamiltonians~\cite{WBHS11,PQS+11}.  Recently a new paradigm has been proposed that uses quantum walks rather than product formulas~\cite{Chi09,BC09}.  This approach is superior to the product formula approach for simulating sparse time-independent Hamiltonians with constant accuracy (and in addition, can be applied to non-sparse Hamiltonians), whereas the product formula approach is superior for generating highly accurate simulations of sparse Hamiltonians.

The performance of simulation algorithms based on product formulas is limited by the fact that high-order approximations are needed to optimize the algorithmic complexity.  The best known high-order product formulas, the Lie--Trotter--Suzuki formulas, approximate the time evolution using a product of unitary operations whose length scales exponentially with the order of the formula \cite{Suz91}.  In contrast, classical methods known as multi-product formulas require a sum of only polynomially many unitary operations to achieve the same accuracy~\cite{BCR99} (although of course the overall cost of classical simulations based on multi-product formulas remains exponential in the number of qubits used to represent the Hilbert space).  However, these methods cannot be directly implemented on a quantum computer because unitary operations are not closed under addition.

Our work addresses this by presenting a non-deterministic algorithm that can be used to perform linear combinations of unitary operators on quantum computers.  We achieve high success probabilities provided the operators being combined are near each other.
We apply this tool to quantum simulation and thereby improve upon existing quantum algorithms for simulating Hamiltonian dynamics.  Our main result is as follows.

\begin{theorem}\label{thm:mainresult}
Let the system Hamiltonian be $H=\sum_{j=1}^m H_j$ where each $H_j\in \C^{2^n\times 2^n}$ is Hermitian and satisfies $\norm{H_j} \le h$ for a given constant $h$. Then the Hamiltonian evolution $e^{-iHt}$ can be simulated on a quantum computer with failure probability and error at most $\epsilon$ as a product of linear combinations of unitary operators.  In the limit of large $m,ht,1/\epsilon$, this simulation uses
\begin{equation}\label{eq:mainresult}
\tilde O\left(m^2 h te^{1.6\sqrt{\log(m h t/\epsilon)}}\right)
\end{equation}
elementary operations and exponentials of the $H_j$s.
\end{theorem}

Although we have not specified the method used to simulate the exponential of each $H_j$, there are well-known techniques to simulate simple Hamiltonians.  In particular, if $H_j$ is $1$-sparse (i.e., has at most one non-zero matrix element in each row and column), then it can be simulated using $O(1)$ elementary operations~\cite{ATS03,CCDFGS03}, so \eq{mainresult} gives an upper bound on the complexity of simulating sparse Hamiltonians.

Our simulation is superior to the previous best known simulation algorithms based on product formulas.  Previous methods have scaling of the same form, but with the coefficient 1.6 replaced by 2.54 \cite[Theorem 1]{BACS07} or 2.06 \cite[Theorem 1]{WBHS10}.
Also note that Theorem 1 of~\cite{PaZh10} gives a similar scaling as in \cite{WBHS10}, except the term in the exponential depends on the second-largest $\norm{H_j}$ rather than $h$.

Perhaps more significant than the quantitative improvement to the complexity of Hamiltonian simulation is that our approach demonstrates a new class of simulation protocols going beyond the Lie--Trotter--Suzuki paradigm, the approach used in most previous simulation algorithms.  It remains unknown how efficiently one can perform quantum simulation as a function of the allowed error $\epsilon$, and we hope our work will lead to a better understanding of this question.

The remainder of this article is organized as follows.
In \sec{add}, we provide a general method for implementing linear combinations of unitary operators using quantum computers and lower bound its success probability.
This method is optimal among a large class of such protocols, as shown in the appendix.
In \sec{mpf}, we provide a brief review of Lie--Trotter--Suzuki formulas and multi-product formulas and then show how to implement multi-product formulas on quantum computers.
Error bounds and overall success probabilities of our simulations are derived in \sec{mpferror}.  We then bound the number of quantum operations used in our simulation in \lem{nexp}, from which \thm{mainresult} follows.
We conclude in \sec{conclusions} with a summary of our results and a discussion of directions for future work.


\section{Adding and Subtracting Unitary Operations Using Quantum Computers}\label{sec:add}

In this section we describe basic protocols for implementing linear combinations of unitary operations.  \lem{main:2case} shows that a quantum computer can nearly deterministically perform a weighted average of two nearby unitary operators.  (Our approach is reminiscent of a technique for implementing fractional quantum queries using discrete queries \cite{CGMSY09}.)  We build upon \lem{main:2case} in \thm{main:nonunitary}, showing that a quantum computer can non-deterministically implement an arbitrary linear combination of a set of unitary operators.

\begin{lemma}\label{lem:main:2case}
Let $U_a,U_b \in \C^{2^n\times 2^n}$ be unitary operations and let $\Delta=\norm{U_a-U_b}$.  Then for any $\kappa\ge 0$, there exists a quantum algorithm that can implement an operator proportional to $\kappa U_a + U_b$ with failure probability at most $\Delta^2\kappa/(\kappa+1)^2\le 4\kappa/(\kappa+1)^2$.
\end{lemma}
\begin{proof}
Let 
\begin{equation}
V_\kappa\defeq
\begin{pmatrix}\sqrt{\frac{\kappa}{\kappa+1}} & \frac{-1}{\sqrt{\kappa+1}}\\
\frac{1}{\sqrt{\kappa+1}} & \sqrt{\frac{\kappa}{\kappa+1}}\end{pmatrix}.
\end{equation}
Our protocol for implementing the weighted average of $U_a$ and $U_b$ works as follows (see \fig{main:addition}).  First, we perform $V_{\kappa}$ on an ancilla qubit.  Second, we perform a zero-controlled $U_a$ gate and a controlled $U_b$ gate on the state $\ket{\psi}$ using the ancilla as the control.  Finally, we apply $V_\kappa^\dagger$ to the ancilla qubit and measure it in the computational basis.  This protocol performs the following transformations:
\begin{align}
\ket{0}\ket{\psi}
&\mapsto
\left(\sqrt{\frac{\kappa}{\kappa+1}}\ket{0}+\frac{1}{\sqrt{\kappa+1}}\ket{1} \right)\ket{\psi}\nonumber\\&\mapsto
\left(\sqrt{\frac{\kappa}{\kappa+1}}\ket{0}U_a\ket{\psi}+\frac{1}{\sqrt{\kappa+1}}\ket{1}U_b\ket{\psi} \right)\nonumber\\
&\mapsto \ket{0}\left(\frac{\kappa}{\kappa+1} U_a + \frac{1}{\kappa+1} U_b \right)\ket{\psi}+\ket{1}\frac{\sqrt{\kappa}}{\kappa+1}(U_b-U_a)\ket{\psi}.\label{eq:thm1:last}
\end{align}
If the first qubit is measured and a result of $0$ is observed, then this protocol performs
 $\ket{\psi}\mapsto (\kappa U_a+U_b)\ket{\psi}$ (up to normalization).  If the measurement yields $1$ then the algorithm fails.  The probability of this failure, $P_+$, is
\begin{equation}
P_+\le \frac{\norm{U_b-U_a}^2\kappa}{(\kappa+1)^2} = \frac{\Delta^2\kappa}{(\kappa+1)^2}.
\end{equation}
Since $\Delta \le 2$, this is at most $\frac{4\kappa}{(\kappa+1)^2}$.
\end{proof}

By substituting $U_b\rightarrow -U_b$, \lem{main:2case} also shows that unitary operations can be subtracted.  (Alternatively, replacing $V_\kappa^\dagger$ with $V_\kappa$ in \fig{main:addition} also simulates $U_a-U_b$.)  Similarly, we could make the weights of each unitary complex by multiplying $U_0$ and $U_1$ by phases, although we will not need to make use of this freedom.

\begin{figure}
\capstart
\[ \large
\Qcircuit @C=1.5em @R=1em @! {
\lstick{|0\rangle} & \gate{V_\kappa} & \ctrlo{1} & \ctrl{1} & \gate{V_\kappa^\dag} & \meter \\
\lstick{|\psi\rangle} & \qw & \gate{U_a} & \gate{U_b} & \qw & \qw
}
\]
\caption{Quantum circuit for non-deterministically performing an operator proportional to $\kappa U_a + U_b$ given a measurement outcome of zero.\label{fig:main:addition}}
\end{figure}
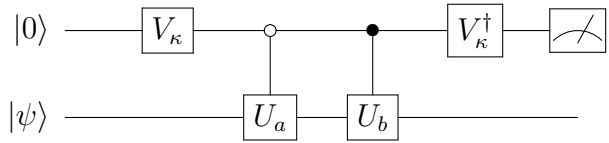

General linear combinations of unitary operators can be performed by iteratively applying \lem{main:2case}.  The following theorem gives constructs such a simulation and provides bounds on the probabilities of failure.

\begin{theorem}\label{thm:main:nonunitary}
Let $V\defeq\sum_{q=1}^{k+1} C_q U_q$ for $k\ge 1$ where $C_q\ne 0$, $\norm{U_q}=1$, and $\max_{q\ne q'}\norm{U_q-U_{q'}} \le \Delta$.  Let $\kappa\defeq (\sum_{q:~C_q>0} C_q)/(\sum_{q:~C_q<0} |C_q|)$. Then there exists a quantum algorithm that implements an operator proportional to $V$ with probability of failure $P_+ + P_-$ with
\begin{align}
  P_+&\le \frac{k\Delta^2}{4}, \\
  P_-&\le \frac{4\kappa}{(\kappa+1)^2},
\end{align}
where $P_+$ is the probability of failing to add some pair of operators and $P_-$ is the probability of failing to perform the subtraction.
\end{theorem}

\begin{proof}
Let
\begin{align}
A&\defeq\frac{1}{\sum_{q:C_q>0} C_q}\sum_{q\colon C_q>0}C_q U_q,\\
B&\defeq\frac{1}{\sum_{q:C_q<0} |C_q|}\sum_{q\colon C_q<0}|C_q| U_q.
\end{align}
We implement $V\propto \kappa A-B$ using circuits that non-deterministically implement operators proportional to $A$ and $B$.  We recursively perform these sums  by using the circuits given in \lem{main:2case}, except that we defer measurement of the output until the algorithm is complete.
We then implement $\kappa A -B$ by using our addition circuit, with $U_a$ taken to be the circuit that implements $A$ and $U_b$ taken to be the circuit that implements $-B$.  Then we measure each of the control qubits for the addition steps.  Finally, if all the measurement results are zero (indicating success),  we measure the control qubit for the subtraction step.  If that qubit is zero, then we know from prior analysis that the non-unitary operation $V$ is implemented successfully.

The operators $A$ and $B$ are implemented by recursively adding terms.  For example, the sum $U_1+U_2+U_3+U_4$ is implemented as $(((U_1+U_2)+U_3)+U_4)$.  Implementing the sums in $A$ and $B$ requires $k-1$ addition operations, so $V$ can be implemented using $k-1$ addition operations and one subtraction operation (assuming that all the control qubits are measured to be zero at the end of the protocol).

According to \lem{main:2case}, the probability of failing to implement $V$, given that we successfully implement $A$ and $B$, is
\begin{equation}
P_-\le \frac{4\kappa}{(\kappa+1)^2}.
\end{equation}
The probability of failing to perform the $k-1$ sums needed to construct $A$ and $B$ obeys
\begin{equation}
P_+
\le (k-1)\frac{\Delta^2\kappa}{(\kappa+1)^2}\le \frac{k\Delta^2}{4},
\end{equation}
where the last step follows from the fact that we can take $\kappa \ge 1$ without loss of generality.
\end{proof}

These results show that we can non-deterministically implement linear combintations of unitary operators with high probability provided that $\kappa\gg 1$ and $\Delta\ll 1$.  As we will see shortly, this situation can naturally occur in quantum simulation problems.

Note that it is not possible to increase the success probability of the algorithm by replacing the single-qubit unitary $V_\kappa$ with a different unitary, even if that unitary is allowed to act on all of the ancilla qubits simultaneously.  We present this argument in \app{opt}.


\section{Implementing Multi-Product Formulas on Quantum Computers}\label{sec:mpf}

In this section we present a new approach to quantum simulation: we approximate the time evolution using a sequence of non-unitary operators that are each a linear combination of product formulas.  Such sums of product formulas are known in the numerical analysis community as multi-product formulas, and can be more efficient than product formulas for classical computations~\cite{BCR99,Chin10}.  We show how to implement multi-product formulas using quantum computers by leveraging our method for non-deterministically performing linear combinations of unitary operators.

\subsection{Review of Lie--Trotter--Suzuki and Multi-Product Formulas}
Product formula approximations can be used to accurately approximate an operator exponential as a product of operator exponentials that can be easily implemented.  Apart from their high degree of accuracy, these approximations are useful because they approximate a unitary operator with a sequence of unitary operators, making them ideally suited for quantum computing applications.

The most accurate known product formula approximations are the Lie--Trotter--Suzuki formulas, which approximate $e^{-iHt}$ for $H=\sum_{j=1}^m H_j$ as a product of the form
$$
  e^{-iHt}\approx\prod_{k=1}^{N_{\exp}} e^{-i H_{j_k} t_k}.
$$
These formulas are recursively defined for any integer $\chi>0$ by~\cite{Suz91}
\begin{align}
S_1(t)&=\prod_{j=1}^m e^{-i H_j t/2}\prod_{j=m}^1 e^{-i H_j t/2},\nonumber\\
S_\chi(t)&=\left(S_{\chi-1}(s_{\chi-1}t)\right)^2S_{\chi-1}([1-4s_{\chi-1}]t)\left(S_{\chi-1}(s_{\chi-1}t)\right)^2\label{eq:1_2},
\end{align}
where $s_p=(4-4^{1/(2p+1)})^{-1}$ for any integer $p>0$.
This choice of $s_p$ is made to ensure that the Taylor series of $S_\chi$ matches that of $e^{-iHt}$ to $O(t^{2\chi +1})$.
Consequently, the approximation can be made arbitrarily accurate for suitably large values of $\chi$ and small values of $t$.

The advantage of these formulas is clear: they are highly accurate and approximate $U(t) \defeq e^{-iHt}$ by a sequence of unitary operations, which can be directly implemented using a quantum computer.  The primary disadvantage is that they require $O(5^k)$ exponentials to construct an $O(t^{2k+1})$ approximation.  This scaling leads to quantum simulation algorithms with complexity $(\norm{H}t)^{1+o(1)}$.  A product formula requiring significantly fewer exponentials could result in a substantial performance improvement over existing product formula-based quantum simulation algorithms.

In the context of classical simulation, multi-product approximations were introduced to address these problems~\cite{BCR99,Chin10}.  Multi-product formulas generalize the approximation-building procedure used to construct $S_\chi$ to allow sums of product formulas.  The resulting formulas are simpler since it is easier to construct a Taylor series by adding polynomials than by multiplication alone.  Specifically, multi-product formulas only need $O(k^2)$ exponentials to construct an approximation of $U(t)$ to $O(t^{2k+1})$.

We consider multi-product formulas of the form
\begin{equation}
M_{k,\chi}(t)=\sum_{q=1}^{k+1} C_q S_\chi(t/\ell_q)^{\ell_q},\label{eq:mkdef}
\end{equation}
where the formula is accurate to $O(t^{2(k+\chi)+1})$.
Here $\ell_1,\ldots,\ell_{k+1}$ are distinct natural numbers and $C_1,\ldots,C_{k+1} \in \R$ satisfy $\sum_{q=1}^{k+1} C_q=1$.  Explicit expressions for the coefficients $C_q$ are known for the case $\chi=1$~\cite{Chin10}:
\begin{equation}
C_q=\prod_{j=\{1,\ldots,k+1\}\setminus q} \frac{\ell_q^2}{\ell_q^2-\ell_j^2}.\label{eq:cqgen}
\end{equation}
We will show later that the same expressions for $C_q$ also apply for $\chi>1$.

For classical simulations, the most numerically efficient formulas correspond to $\ell_q=q$ and $\chi=1$.  The simplest example of a multi-product formula is the Richardson extrapolation formula for $S_1$~\cite{Ri1911}:
\begin{equation}
U(t)=\frac{4S_1(t/2)^2-S_1(t)}{3}+O(t^5).
\end{equation}
Even though the above expression is non-unitary, it is very close to a unitary operator.  In fact, there exists a unitary operator within distance $O(t^{10})$ of the multi-product formula.  In general, Blanes, Casas, and Ros show that if a multi-product formula $M_{k,\chi}$ is accurate to $O(t^{2(k+\chi)+1})$ then it is unitary to $O(t^{4(k+\chi)+2})$~\cite{BCR99}; therefore such formulas are practically indistinguishable from unitary operations in many applications~\cite{Chin10,BCR99}.

The principal drawback of these formulas is that they are less numerically stable than Lie--Trotter--Suzuki formulas.  Because they involve sums that nearly perfectly cancel, substantial roundoff errors can occur in their computation.  Such errors can be mitigated by using high numerical precision or by summing the multi-product formula in a way that minimizes roundoff error.

An additional drawback is that linear combinations of unitary operators are not natural to implement on a quantum computer.  Furthermore, our previous discussion alludes to a sign problem for the integrators: the more terms of the multi-product formula that have negative coefficients, the lower the success probability of the implementation of \thm{main:nonunitary}. This sign problem cannot be resolved completely because, as shown by Sheng~\cite{She89}, it is impossible to construct a high-order multi-product formula of the form in~\eq{mkdef} without using negative $C_q$.  Nevertheless, we show that this problem is not fatal and that multi-product formulas can be used to surpass what is possible with previous simulations based on product formulas.

\subsection{Implementing Multi-Product Formulas Using Quantum Computers}
We now discuss how to implement multi-product formulas using quantum computers.  The main obstacle is that the multi-product formulas most commonly used in classical algorithms have a value of $\kappa$ (as defined in \thm{main:nonunitary}) that approaches $1$ exponentially quickly as $k$ increases.  Thus the probability of successfully implementing such multi-product formulas using \thm{main:nonunitary} is exponentially small.

Instead, we seek a multi-product formula $M_{k,\chi}$, with a large value of $\kappa$, such that $\norm{M_{k,\chi}(\lambda)-U(\lambda)} \in O(\lambda^{2(k+\chi)+1})$.  Although many choices are possible, we take our multi-product formulas to be of the following form because they yield a large value of $\kappa$ while consisting of relatively few exponentials.
\begin{definition}\label{def:MPF}
Let $k \ge 0$ and $\chi \ge 1$ be integers, $\gamma$ a real number such that $e^{\gamma (k+1)}$ is an integer, and $S_\chi(\lambda)$ a symmetric product formula approximation to $U(\lambda)$ obeying $\norm{S_\chi(\lambda)-U(\lambda)} \in O(\lambda^{2\chi+1})$.  Then for any $t\in\R$ we define the multi-product formula $M_{k,\chi}(t)$ as
\begin{equation}
M_{k,\chi}(t) \defeq \sum_{q=1}^{k+1} C_q S_\chi(t/\ell_q)^{\ell_q}\label{eq:MkDef}
\end{equation}
where
\begin{equation}
\ell_q\defeq\begin{cases}q&\textrm{if $q\le k$}\\
e^{\gamma (k+1)}& \textrm{if $q=k+1$}\end{cases}\label{eq:lqDef}
\end{equation}
and
\begin{equation}
C_q\defeq \begin{cases}\frac{q^2}{q^2-e^{2\gamma (k+1)}} \prod_{j\ne q}^{k} \frac{q^2}{q^2-j^2}&\text{if $q\le k$}\\ \prod_{j=1}^{k} \frac{e^{2\gamma (k+1)}}{e^{2\gamma (k+1)}-j^2}& \text{if $q=k+1$}.\end{cases} \label{eq:cq}
\end{equation}
\end{definition}
We choose these values of $\ell_q$ because for sufficiently large $\gamma$ they guarantee that $C_{k+1}$ is much larger in absolute value than all other coefficients.  This ensures a high success probability in \thm{main:nonunitary} because $\kappa \ge |C_{k+1}|/\sum_{q=1}^k |C_{q}|$ is
large if $|C_{k+1}|$ exceeds the sum of all other $|C_q|$.

The following lemma shows that $M_{k,\chi}$ is a higher-order integrator than $S_k$.  Quantitative error bounds are proven in the next section.

\begin{lemma}\label{lem:MPFHigherOrder}
Let $M_{k,\chi}$ be a multi-product formula constructed according to \defn{MPF}.  Then for $\lambda\ll 1$ we have
\begin{equation}
\norm{M_{k,\chi}(\lambda)-U(\lambda)} \in O(\lambda^{2(k+\chi)+1}).
\end{equation}
\end{lemma}

\begin{proof}
We follow the steps outlined in Chin's proof for the case where $\chi=1$~\cite{Chin10}.  As shown in~\cite{BCR99}, a sufficient condition for a multi-product formula of the form $M_{k,\chi}(\lambda)=\sum_p C_p S_\chi(\lambda/\ell_p)^{\ell_p}$ to satisfy $\norm{U(\lambda)-M_{k,\chi}(\lambda)}\in O(\lambda^{2(k+\chi)+1})$ is for $C$ to satisfy the following matrix equation:
\begin{equation}
\begin{pmatrix}
1&1&1&\cdots &1\\
\ell_1^{-2\chi}&\ell_2^{-2\chi}&\ell_3^{-2\chi}&\cdots&\ell_{k+1}^{-2\chi}\\
\ell_1^{-2\chi-2}&\ell_2^{-2\chi-2}&\ell_3^{-2\chi-2}&\cdots&\ell_{k+1}^{-2\chi-2}\\
\vdots&\vdots&\vdots&\ddots&\vdots\\
\ell_1^{-2(k+\chi-1)}&\ell_2^{-2(k+\chi-1)}&\ell_3^{-2(k+\chi-1)}&\cdots&\ell_{k+1}^{-2(k+\chi-1)}
\end{pmatrix}
\begin{pmatrix}C_1\\C_2\\C_3\\\vdots\\C_{k+1}\end{pmatrix}
= \begin{pmatrix} 1\\0\\0\\\vdots\\0\end{pmatrix}.
\end{equation}
This ensures that the sum of all $C_q$ is $1$ and the coefficients of all the error terms in the multi-product formula are zero up to $O(\lambda^{2(k+\chi)-1})$. Denoting the matrix in the above equation by $V$,  the vector $C$ is then the first column of $V^{-1}$. The matrix $V$ is a generalized Vandermonde matrix, which can be explicitly inverted~\cite{Moa03}.  The entries of $C$ correspond to~\eq{cqgen}, which  coincides with the values of $C_q$ given in~\eq{cq} for the values of $\ell_q$ in~\eq{lqDef}.  Therefore the result of \cite{BCR99} shows that these values of $C_q$ extrapolate an $O(\lambda^{2\chi+1})$ symmetric product formula into an $O(\lambda^{2(k+\chi)+1})$ multi-product formula, as claimed.
\end{proof}

The following upper bound on the coefficients of the multi-product formula will be useful.

\begin{lemma}\label{lem:cqbound}
If $e^{2\gamma (k+1)} \ge 2k^2$, then for all $1 \le q < k+1$, the coefficients $C_q$ from \defn{MPF} satisfy
\begin{align}
|C_q|\le \sqrt{2} \, k^{3/2} e^{2k(1+\log(\eta)/2-\gamma)}
\label{eq:cqexpansion8}
\end{align}
where
\begin{equation}
\eta \defeq \max_{\lambda\in [0,1)} \frac{\lambda^2}{(1+\lambda)^{1+\lambda}(1-\lambda)^{1-\lambda}} \approx 0.3081.
\label{eq:eta}
\end{equation}
\end{lemma}

\begin{proof}
For any $q<k+1$, \defn{MPF} gives
\begin{align}
C_q
&= \frac{q^2}{q^2-e^{2\gamma (k+1)}} \prod_{j\in \{1,\ldots,k\}\setminus q} \frac{q^2}{q^2-j^2} \nn
&= \frac{q^2}{q^2-e^{2\gamma (k+1)}} \prod_{j\in \{1,\ldots,k\}\setminus q} \frac{q^2}{(q+j)(q-j)} \nn
&= \frac{q^2}{q^2-e^{2\gamma (k+1)}} q^{2(k-1)} \frac{2q \cdot q!}{(q+k)!} \frac{(-1)^{k-q}}{(q-1)!(k-q)!} \nn
&= \frac{(-1)^{k-q}2q^{2k+2}}{(k+q)!(k-q)! (q^2-e^{2\gamma (k+1)})}\label{eq:cqexpansion2}.
\end{align}
Using $e^{2\gamma (k+1)} \ge 2k^2 \ge 2q^2$, we have the bound
\begin{align}
|C_q| &\le \frac{4q^{2k+2}e^{-2\gamma (k+1)}}{(k+q)!(k-q)!}\label{eq:cqexpansion3}.
\end{align}

We proceed by using the lower bound \cite{AbSt64}
\begin{equation}
n!\ge \sqrt{2\pi n}\,n^{n} e^{-(n+1/13)}\label{eq:factorialbd}.
\end{equation}
We will use this bound differently to estimate $|C_q|$ for the cases where $q<k$ and $q=k$.  Using~\eq{factorialbd}, we
lower bound both factorial functions in~\eq{cqexpansion3} as follows:
\begin{align}
|C_q|  &\le \frac{4q^{2k+2}e^{-2\gamma(k+1)+2k+2/13}}{2\pi \sqrt{k}(k+q)^{k+q}(k-q)^{k-q}}
\label{eq:cqexpansion4}.
\end{align}
Here we have used the fact that $\sqrt{(k+q)(k-q)}\ge \sqrt{k}$ for $q\le k-1$.
Now introduce a parameter $\lambda$ such that $q=k\lambda$.  We simplify~\eq{cqexpansion4} using this substitution and
divide the numerator and denominator of~\eq{cqexpansion4} by $k^{2k}$ to find
\begin{align}
|C_q| &\le \frac{2k^{3/2}e^{2k(1-\gamma)-2\gamma+2/13}}{\pi}\left(\frac{\lambda^{2}}{(1+\lambda)^{1+\lambda}(1-\lambda)^{1-\lambda}} \right)^k\nonumber\\
&\le \frac{2k^{3/2}e^{2k(1-\gamma)-2\gamma+2/13}}{\pi}\eta^k
\label{eq:cqexpansion5}.
\end{align}
We then simplify this expression and find that
\begin{equation}
|C_q|\le\frac{2k^{3/2}e^{2k(1+\log(\eta)/2 -\gamma)}}{\pi e^{2\gamma-2/13}},
\end{equation}
which for $\gamma>0$ is bounded above by
\begin{align}
|C_q|\le {{k^{3/2}e^{2k(1+\log(\eta)/2-\gamma)}}}
< \sqrt{2} \, {{k^{3/2}e^{2k(1+\log(\eta)/2-\gamma)}}}.\label{eq:cqexpansion7}
\end{align}

Our bound for the case where $q=k$ is found using similar (but simpler) reasoning.  We first substitute $q=k$ into~\eq{cqexpansion3},
use the lower bound in~\eq{factorialbd} to remove the factorial function, and simplify the result to find
\begin{equation}
|C_{k}|\le 2{k^{3/2}e^{2k(1-\log(2)-\gamma)}}/\sqrt{\pi}
< \sqrt{2} \, {k^{3/2}e^{2k(1-\log(2)-\gamma)}},
\end{equation}
which is less than the value in~\eq{cqexpansion7} because $-0.6931 \approx -\log 2 < \log (\eta)/2\approx -0.5886$.
Therefore \eq{cqexpansion8} holds for all $q\le k$ as required.
\end{proof}

The value of $\gamma$ used in $M_{k,\chi}$ can be chosen to minimize the probability of a subtraction error.  The following lemma relates the value of $\kappa$ to $\gamma$, allowing us to use \thm{main:nonunitary} to find a value of $\gamma$ that ensures a sufficiently small probability of a subtraction error.

Henceforth we assume that $k\ge 1$.  This is because $k=0$ corresponds to an ordinary product formula and the bounds that we prove for multi-product formulas can tightened by excluding this case, which is also already well analyzed~\cite{BACS07,Suz91}.

\begin{lemma}\label{lem:1}
Let $M_{k,\chi}$ be a multi-product formula as in \defn{MPF} and let $\kappa$ be defined for $M_{k,\chi}$ as in \thm{main:nonunitary}. Then if $2k^2 \le e^{2\gamma (k+1)}$ and $k>0$, we have
\begin{equation}
\kappa \ge 2^{-1/2}e^{-2k(1+\log(\eta)/2-\gamma)-\log(k^{5/2})}\label{eq:lem1:kappabound}
\end{equation}
where $\eta$ is defined in \eq{eta}.
\end{lemma}

\begin{proof}
According to \thm{main:nonunitary}, we have
$
\kappa = {\Sigma_+}/{\Sigma_-}
$,
where $\Sigma_+$ is the sum of all $C_q$ with positive coefficients and $\Sigma_-$ is the absolute value of the corresponding negative sum.  A lower bound on $\kappa$ is therefore found by dividing a lower bound on $\Sigma_+$ by an upper bound on $\Sigma_-$.

Using the expression for $C_q$ from \defn{MPF}, we have
\begin{align}
C_{k+1}
= \prod_{j=1}^{k} \frac{e^{2\gamma (k+1)}}{e^{2\gamma (k+1)}-j^2}
= \prod_{j=1}^{k} \frac{1}{1-j^2e^{-2\gamma (k+1)}}.
\label{eq:lem5:rhs19}
\end{align}
The denominators on the right hand side of~\eq{lem5:rhs19} are positive under the assumption that $2k^2\le e^{2\gamma (k+1)}$, which ensures that $k^2e^{-2\gamma (k+1)}<1$ and simplifies the subsequent results of \cor{gamma}.  We also have $C_{k+1}\ge 1$ because each denominator is less than $1$.  Since $C_{k+1} > 0$, we have $\Sigma_+\ge C_{k+1} \ge 1$, and therefore
\begin{equation}
\kappa \ge \frac{1}{\Sigma_-}.\label{eq:lem2:sigmaminus}
\end{equation}

Next we provide an upper bound for $\Sigma_-$.  Since $C_{k+1} > 0$, we have
\begin{equation}
\Sigma_- \le \sum_{q=1}^{k} |C_q|.\label{eq:lem5:sigma-bd}
\end{equation}
An upper bound for $\Sigma_-$ can then be obtained directly from upper bounds for $\max_{q<k+1}|C_q|$.

Using \lem{cqbound}, we have
\begin{equation}
\Sigma_-\le \sum_{q=1}^{k} {{\sqrt{2}\,k^{3/2}e^{2k(1+\log(\eta)/2-\gamma)}}}< \sqrt{2}\,k^{5/2}e^{2k(1+\log(\eta)/2-\gamma)}.\label{eq:lem2:sumabsCq}
\end{equation}
We substitute this inequality into~\eq{lem2:sigmaminus} to obtain
\begin{equation}
\kappa\ge 2^{-1/2}k^{-5/2}e^{-2k(1+\log(\eta)/2-\gamma)}
\end{equation}
as claimed.
\end{proof}

In fact, the bound of \lem{1} is nearly tight, in that $\kappa$ decays exponentially with $k$ if $\gamma < 1+\log(\eta)/2$, as we discuss in more detail below.  Thus the success probability of our algorithm decays exponentially if $\gamma$ is too small.  Consequently, we will find that unlike the classical case, our quantum algorithm does not provide poly-logarithmic error scaling.

The following corollary provides a sufficient value of $\gamma$ to ensure that the probability of our algorithm making a subtraction error is small. 

\begin{corollary}\label{cor:gamma}
Let $M_{k,\chi}$ be a multi-product formula as in \defn{MPF}, let $\kappa$ be defined for $M_{k,\chi}$ as in \thm{main:nonunitary}, and let $\delta\le 1$.  Furthermore, suppose $k> 0$ and
$$
  \gamma\ge 1+\frac{\log(\eta)}{2}+\frac{1}{2k} \log\left(\frac{(2k)^{\frac{5}{2}}}{\delta}\right).
$$
Then $P_-\le \delta$ and $k^2e^{-2\gamma (k+1)}\le 1/2$.
\end{corollary}

\begin{proof}
Without loss of generality, we can take $\kappa\ge 1$ for our subtraction step because $\sum_q C_q =1$ and hence $\kappa=\Sigma_+/\Sigma_- \ge 1$.  This observation and
the result of \thm{main:nonunitary} imply that the probability of failing to perform the subtraction step in our implementation of $M_{k,\chi}$ satisfies
\begin{equation}
P_-\le 4/\kappa.\label{eq:P-bd}
\end{equation}
Eq.~\eq{P-bd} and the bounds on $\kappa$ in \lem{1} give $P_-\le \delta$ provided
\begin{equation}
4\sqrt{2}\,e^{2k(1+\log(\eta)/2-\gamma)+\log(k^{5/2})}\le \delta.\label{eq:lemgamma:errbd}
\end{equation}
We obtain our sufficient value of $\gamma$ by solving~\eq{lemgamma:errbd} for $\gamma$, giving
\begin{equation}
\gamma\ge 1+\log(\eta)/2+\frac{1}{2k} \log\left(\frac{(2k)^{\frac{5}{2}}}{\delta}\right).\label{eq:gammabd}
\end{equation}
Eq.~\eq{gammabd} also implies that $k^2e^{-2\gamma (k+1)}\le 1/2$.  For $0<\delta\le 1$ and $\gamma$ saturating~\eq{gammabd}, it is easy to see that $k^2 e^{-2\gamma (k+1)}$ is a monotonically decreasing function of $k$, and therefore achieves its maximum value at $k=1$, the smallest possible value of $k$.  We find that $k^2e^{-2\gamma (k+1)}<1/2$ at $k=1$ for $\delta=1$, and therefore the condition $k^2e^{-2\gamma (k+1)}\le 1/2$ is automatically implied by our choice of $\gamma$ in~\eq{gammabd}.
\end{proof}

\begin{figure}
  \capstart
	\includegraphics[width=0.65\linewidth]{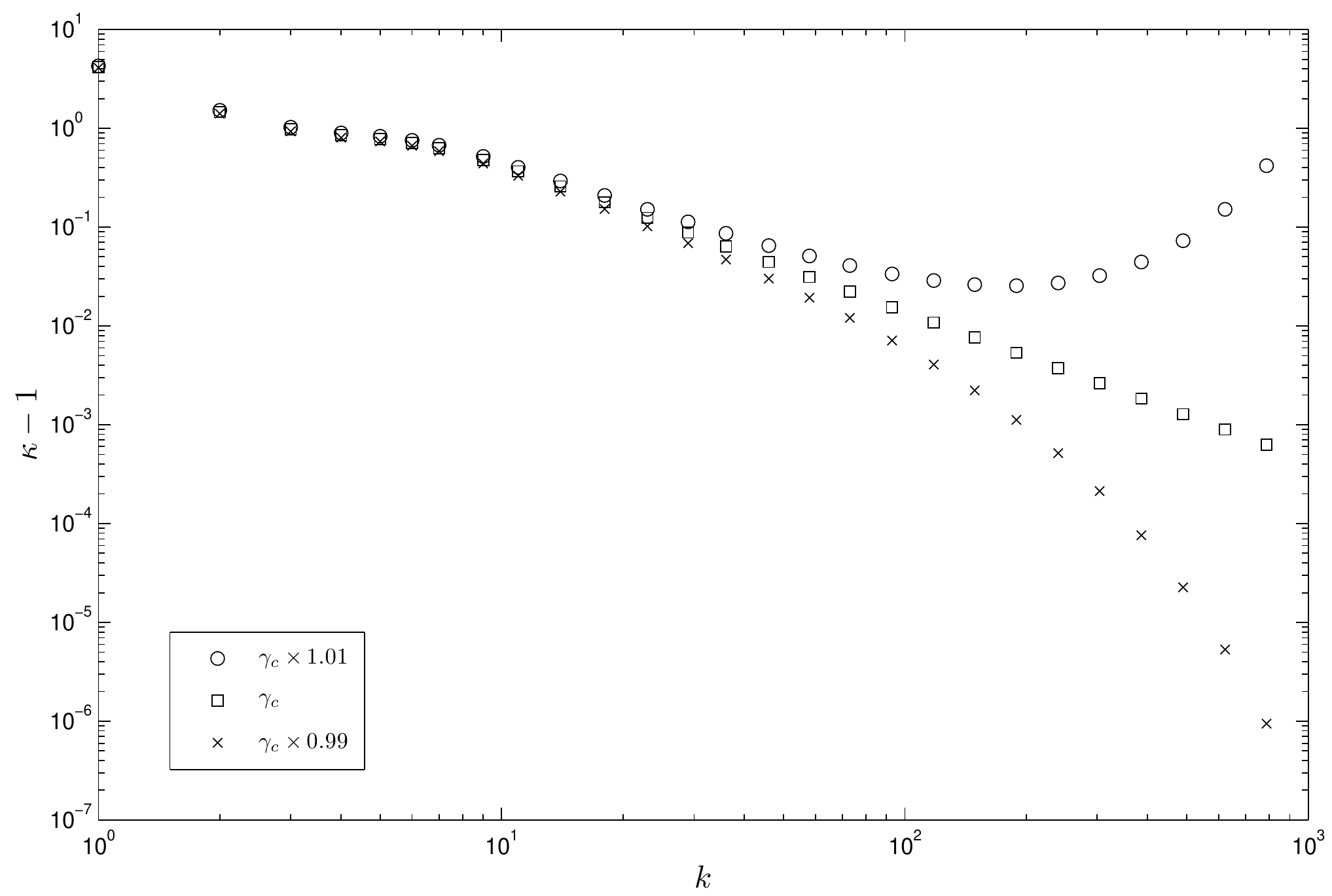}
	\caption{Scaling of $\kappa$ with $k$ for three values of $\gamma$ centered around $\gamma_c \defeq 1+\log(\eta)/2$.  The data show that $\kappa$ approaches $1$ polynomially quickly if $\gamma=\gamma_c$, whereas a slight increase in $\gamma$ causes $\kappa$ to grow exponentially and a slight decrease causes $\kappa$ to converge to $1$ exponentially with $k$.
	}
\label{fig:etadata}
\end{figure}

The value of $\gamma$ given by \cor{gamma} is tight up to $O(k^{-1} \log k)$.  This is illustrated in \fig{etadata}, which shows $\kappa$ as a function of $k$ for $\gamma=\gamma_c\defeq 1+\log(\eta)/2$ and two slightly perturbed values of $\gamma$ centered around $\gamma_c$. We see that small deviations away from $\gamma=\gamma_c$ lead to either exponential growth of $\kappa$ or exponential convergence of $\kappa$ to $1$.  Thus our lower bound for $\gamma$ cannot be significantly improved.


\section{Analysis of Simulation and Errors in Multi-product Formulas\label{sec:mpferror}}

The results of the previous section show how $\kappa$ scales with the number of terms in the multi-product formula used in the simulation. We now expand on these results by bounding the approximation errors incurred by using the multi-product formula.  We also present an error correction method to ensure that our implementation fails at most a constant fraction of the time.  Our error bounds are established as follows.  First, we estimate the error in multi-product formulas that utilize high-order Lie--Trotter--Suzuki formulas.  Second, we discuss how these multi-product formulas are implemented.  Finally, we estimate the inversion error for the resulting multi-product formulas and bound the average error resulting from a given step.

\begin{lemma}\label{lem:MPFSuzError}
Let $M_{k,k}(\lambda)$ satisfy \defn{MPF} for evolution time $\lambda\ge 0$, let $|C_q|\le 2$ for all $q=1,\ldots,k+1$, and let $h\lambda\le \frac{3\log(2)}{4mk(5/3)^{k-1}}$.  Then
\begin{equation}
\norm{U(\lambda)-M_{k,k}(\lambda)} \le (2m(5/3)^{k-1}h \lambda)^{4k+1}.
\end{equation}
\end{lemma}

The proof of \lem{MPFSuzError} requires upper bounds on the remainder terms of Taylor series expansions.  Let $\Rem_\ell(f)$ denote the remainder term of the Taylor series of a function $f$ truncated at order $\ell$.  The following lemma bounds the remainder term for an operator exponential.

\begin{lemma}\label{lem:Rl}
Let $a_j \in \R$ for $j=1,\ldots,M$ and suppose $\norm{H_j} \le h$.  Then
\begin{equation}
\Norm{\Rem_{\ell}\left(\prod_{j=1}^M e^{-i a_j H_j t}\right)} \le \frac{(\sum_{j=1}^M |a_j| h t)^{\ell+1}}{(\ell+1)!}\exp\left(\sum_{q=1}^M |a_q| h t\right).
\end{equation}
\end{lemma}

\begin{proof}
Using the triangle inequality and sub-multiplicativity of the norm,
we find
\begin{align}
\Norm{\Rem_{\ell}\left(\prod_{j=1}^M e^{-i a_j H_j t}\right)}
&= \Norm{\Rem_{\ell}\left(\prod_{j=1}^M \left( \sum_{p=0}^\infty (-i a_j H_j t)^p/p!\right)\right)} \\
&\le \Rem_{\ell}\left(\prod_{j=1}^M \left( \sum_{p=0}^\infty (|a_j| h t)^p/p!\right)\right) \nn
&= \Rem_{\ell}\left(\exp\left(\sum_{j=1}^M|a_j| h t \right)\right) \nn
&=\sum_{p=\ell+1}^\infty \frac{\left(\sum_{j=1}^M|a_j| h t \right)^p}{p!}\nonumber\\
&\le\frac{\left(\sum_{j=1}^M|a_j| h t \right)^{\ell+1}}{(\ell+1)!}\exp\left(\sum_{j=1}^M|a_j| h t  \right)
\end{align}
as claimed.
\end{proof}

Now we are ready to prove \lem{MPFSuzError}.

\begin{proofof}{\lem{MPFSuzError}}
\lem{MPFHigherOrder} implies that
\begin{equation}
\norm{M_{k,k}(\lambda)-U(\lambda)} \in O(\lambda^{4k+1}),
\end{equation}
so the approximation error is entirely determined by the terms of order $\lambda^{4k+1}$ in $M_{k,k}$ and $U$.  If we remove the terms in the Taylor series of $M_{k,k}$ and $U$ that cancel, the remainders, $\Rem_{4k}(M_{k,k}(\lambda))$ and $\Rem_{4k}(U(\lambda))$, determine the error via
\begin{align}
\norm{M_{k,k}(\lambda)-U(\lambda)}
&=\norm{\Rem_{4k}(M_{k,k}(\lambda))-\Rem_{4k}(U(\lambda))} \nn
&\le \norm{\Rem_{4k}(M_{k,k}(\lambda))}+\norm{\Rem_{4k}(U(\lambda))}.
\label{eq:lem8:triangle}
\end{align}

\lem{Rl} implies that
\begin{align}
\norm{\Rem_{4k}(U(\lambda))}
&\le \frac{(mh\lambda)^{4k+1}}{(4k+1)!}e^{mh\lambda} \nn
&< \left(\frac{4}{3}m (5/3)^{k-1}  h \lambda\right)^{4k+1}.\label{eq:R4kU}
\end{align}
The second inequality in~\eq{R4kU} follows from the assumption that $h\lambda\le \frac{3\log(2)}{4mk(5/3)^{k-1}}$, which implies that $\exp(mh\lambda)/(4k+1)!\le 2^{3/4}/5!< 1$.

The definition of $M_{k,\chi}$ implies
\begin{align}
\norm{\Rem_{4k}(M_{k,k}(\lambda))}
&\le \sum_{q=1}^{k+1} |C_q|\Norm{\Rem_{4k}(S_k(\lambda/\ell_q)^{\ell_q})}.
\label{eq:lem8:r4kMkk}
\end{align}
Thus we upper bound $\norm{\Rem_{4k}(S_k(\lambda/p)^{p})}$.
This bound follows similar logic to the bound for $U(\lambda)$, but the calculation is slightly more complicated because $S_k$ is the product of many exponentials.  Specifically,
\begin{equation}
S_k(\lambda/p)=\prod_{\ell=1}^{2m5^{k-1}} e^{-iH_{j_\ell} q_{k,\ell} \lambda/p},
\end{equation}
where $q_{k,\ell}$ is the ratio between $\lambda/p$ and the duration of the $\ell^{\rm th}$ exponential in $S_k$.
\lem{Rl} gives
\begin{equation}
\norm{\Rem_{4k}(S_k(\lambda/p)^p)} \le \frac{(2m5^{k-1} \max q_{k,\ell} h \lambda)^{4k+1}}{(4k+1)!}e^{2m5^{k-1} \max q_{k,\ell} h \lambda}.\label{eq:lem8:R4kbd2}
\end{equation}
Using the upper bound $q_{k,\ell}\le {2k}/{3^k}$ from Appendix A of~\cite{WBHS10}, we have
\begin{equation}
\norm{\Rem_{4k}(S_k(\lambda/p)^p)} \le \frac{(\frac{4}{3}mk (5/3)^{k-1}  h \lambda)^{4k+1}}{(4k+1)!} e^{\frac{4}{3}mk (5/3)^{k-1}  h \lambda}.\label{eq:lem8:51}
\end{equation}
This bound can be simplified using $(4k+1)!\ge k^{4k+1}5!$ for $k\ge 1$ (a consequence of~\eq{factorialbd}) and the hypothesis that $\frac{4}{3}mk (5/3)^{k-1}  h \lambda\le \log(2)$, giving
\begin{equation}
\norm{\Rem_{4k}(S_k(\lambda/p)^p)} \le \frac{2}{5!}\left(\frac{4}{3}m (5/3)^{k-1}  h \lambda\right)^{4k+1}.\label{eq:Skremainder}
\end{equation}
Using this result and the assumption that $|C_q|\le 2$ in \eq{lem8:r4kMkk} gives
\begin{equation}
\norm{\Rem_{4k}(M_{k,k}(\lambda))}
\le \frac{8(k+1)}{5!}\left(\frac{4}{3}m (5/3)^{k-1}  h \lambda\right)^{4k+1}\label{eq:lem8:r4kMkk2}.
\end{equation}

Combining~\eq{lem8:triangle},~\eq{R4kU}, and~\eq{lem8:r4kMkk2} gives
\begin{align}
\norm{U(\lambda)-M_{k,k}(\lambda)} &\le \left(1+\frac{8(k+1)}{5!}\right)\left(\frac{4}{3}m (5/3)^{k-1}  h \lambda\right)^{4k+1}\nonumber\\
&\le (2m (5/3)^{k-1}  h \lambda)^{4k+1},
\label{eq:lem8:final}
\end{align}
proving the lemma.
\end{proofof}

A useful consequence of performing $M_{k,k}$ using a single subtraction step is that if a subtraction error occurs, the simulator performs the operation
\begin{equation}
E_k(\lambda) \colon \ket{\psi} \mapsto \frac{\sum_{q} |C_{q}| S_k(\lambda/\ell_q)^{\ell_q}\ket{\psi}}{\norm{\sum_{q} |C_q| S_{k}(\lambda/\ell_q)^{\ell_q}\ket{\psi}}}.
\label{eq:ek}
\end{equation}
This error operation can be approximately corrected because, as Blanes et al.\ proved~\cite[Theorem 1]{BCR99}, \begin{equation}E_k(-\lambda)E_k(\lambda)=\openone+O(\lambda^{4k+2}).\label{eq:blanesinv}\end{equation}
Since the coefficients $|C_q|$ are all positive, \thm{main:nonunitary} shows that the approximate correction operation $E_k(-\lambda)$ can be performed with success probability close to $1$ provided $\Delta$ is small. The following lemma states that the error incurred by approximately correcting subtraction errors is at most equal to our upper bound for the approximation error for $M_{k,k}(\lambda)$.
\begin{lemma}\label{lem:ekinv}
Let $E_k(\lambda)$ act as in \eq{ek}, where $C_q$ and $\ell_q$ are given in \defn{MPF}.  If $2mk(5/3)^{k-1}h\lambda\le 1/2$, then \rm 
\begin{equation}
\max_{\ket{\psi}}\norm{\left(\openone -E_k(-\lambda)E_k(\lambda)\right)\ket{\psi}}
\le \left(2mk(5/3)^{k-1}h \lambda\right)^{4k+2}.
\end{equation}
\end{lemma}

\begin{proof}
By \eq{ek} and \eq{blanesinv},
\begin{align}
\max_{\ket{\psi}}\norm{\left(\openone -E_k(-\lambda)E_k(\lambda)\right)\ket{\psi}}
&= \max_{\ket{\psi}}\norm{\Rem_{4k+1}(E_k(-\lambda)E_k(\lambda)\ket{\psi})}\nonumber\\
&\le \frac{\Norm{\Rem_{4k+1}\left(\sum_p \sum_q |C_p| |C_q| S_k(-\lambda/p)^pS_k(\lambda/q)^q\right)}}{\min_{\ket{\phi}}\norm{\sum_p |C_p| S_k(\lambda/\ell_p)^{\ell_p}\ket{\phi}}^2}.\label{eq:ekinv:sk1}
\end{align}
We then follow the same reasoning used in the proof of \lem{Rl}.
By the triangle inequality, the norm of the remainder of a Taylor series is upper bounded by the sums of the norms of the individual terms in the remainder.  This can be bounded by replacing the exponent of each exponential in $S_k$ with its norm. We use similar reasoning to that used in~\eq{lem8:51} to find
\begin{equation}
\Norm{\Rem_{4k+1}\left(S_k(-\lambda/p)^pS_k(\lambda/q)^q\right)} \le \Rem_{4k+1}\left( e^{\frac{8}{3}mk (5/3)^{k-1}  h \lambda} \right),
\end{equation}
so the numerator of \eq{ekinv:sk1} satisfies
\begin{align}
\Norm{\Rem_{4k+1}\left(\sum_p \sum_q |C_p| |C_q| S_k(-\lambda/p)^pS_k(\lambda/q)^q\right)}
&\le \sum_p \sum_q |C_p| |C_q| \Rem_{4k+1}\left( e^{\frac{8}{3}mk (5/3)^{k-1}  h \lambda}\right) \nn
&\le \norm{C}_1^2 \frac{ \left(\tfrac{8}{3}mk(5/3)^{k-1}h\lambda\right)^{4k+2}e^{\frac{8}{3}mk(5/3)^{k-1}h\lambda}}{(4k+2)!}
\end{align}
where $C$ is a vector with entries $C_p$, so $\norm{C}_1 = \sum_p |C_p|$.
Our assumptions imply $\frac{8}{3}mk(5/3)^{k-1}h\lambda\le 2/3 < \log(2)$, so
\begin{align}
\Norm{\Rem_{4k+1}\left(\sum_p \sum_q |C_p| |C_q| S_k(-\lambda/p)^pS_k(\lambda/q)^q\right)}
&\le \norm{C}_1^2 \frac{2 \left(\frac{8}{3}mk(5/3)^{k-1}h\lambda\right)^{4k+2}}{(4k+2)!} \nn
&\le \norm{C}_1^2 \frac{2 \left(\frac{2e}{3}m(5/3)^{k-1}h\lambda\right)^{4k+2}}{\sqrt{12\pi}e^{25/26}}
\label{eq:ekinv:sk3},
\end{align}
where the last inequality results from using Stirling's approximation as given in~\eq{factorialbd}.

The denominator of \eq{ekinv:sk1} can be lower bounded as follows:
\begin{align}
\min_{\ket{\phi}}{\Norm{\sum_p |C_p| S_k(\lambda/\ell_p)^{\ell_p}\ket{\phi}}}
&= \min_{\ket{\phi}}{\Norm{\sum_p |C_p| (e^{-iHt}-(e^{-iHt}-S_k(\lambda/\ell_p)^{\ell_p}))\ket{\phi}}}\nonumber\\
&\ge \norm{C}_1 \left(1-\max_p \norm{e^{-iHt}-S_k(\lambda/\ell_p)^{\ell_p}}\right)\nonumber\\
&= \norm{C}_1 (1-\norm{e^{-iHt}-S_k(\lambda)}).
\end{align}
Since $2mk(5/3)^{k-1}h\lambda \le 1/2 < 3/(4\sqrt{2})$, Theorem 3 of~\cite{WBHS10} implies that $\norm{e^{-iHt}-S_k(\lambda)} \le 2(2 mk(5/3)^{k-1}h\lambda)^{2k+1}$.
Using $2mk(5/3)^{k-1}h\lambda\le 1/2$ and $k\ge 1$ we have that
\begin{equation}
\norm{e^{-iHt}-S_k(\lambda)} \le \frac{1}{4},
\end{equation}
implying
\begin{equation}
\min_{\ket{\phi}}
\Norm{\sum_p |C_p| S_k(\lambda/\ell_p)^{\ell_p} \ket{\phi}}
\ge \frac{3}{4} \norm{C}_1.
\end{equation}
Combining this with our upper bound on the numerator gives
\begin{align}
\max_{\ket{\psi}}
\norm{(\openone -E_k(-\lambda)E_k(\lambda))\ket{\psi}}
&\le \frac{2(4/3)^2}{\sqrt{12\pi}e^{25/26}}\left(\frac{2e}{3}m(5/3)^{k-1}h\lambda\right)^{4k+2} \nn
&\le (2m(5/3)^{k-1}h\lambda)^{4k+2}\label{eq:ekinv:sk4}
\end{align}
as claimed.
\end{proof}

We simulate $U(t)$ using $r$ iterations of $M_{k,k}(t/r)$ for some sufficiently large $r$.
Our next step is to combine \lem{MPFSuzError} and \lem{ekinv} to find upper bounds on $r$ such that $U(t)$ is approximated to within some fixed error.
We take $\delta=1/2$, i.e., we accept a maximum failure probability of $1/2$ for each multi-product formula.  We then sum the cumulative errors and use the Chernoff bound to show that, with high probability, the simulation error is at most $\epsilon$.  These results are summarized in the following lemma.

\begin{lemma}\label{lem:r}
Let $M_{k,k}$ be a multi-product formula given by \defn{MPF} with $|C_q|\le 2$ for all $q\le k+1$.  Let $\gamma$ be chosen as in \cor{gamma} with $\delta=1/2$ and let the integer $r$ satisfy
\begin{equation}
r\ge \frac{(2m (5/3)^{k-1}h t)^{1+1/4k}}{(\epsilon/5)^{1/4k}}\label{eq:lem11:rbound}
\end{equation}
for $\epsilon\le m h tk^{-4k}$.
Then a quantum computer can approximately implement $U(t)$ as $M_{k,k}(t/r)^r$ with error at most $\epsilon$ and with probability at least $1-e^{-r/13}$, assuming that no addition errors occur during the simulation, while utilizing no more than $5r$ subtraction attempts and approximate inversions.
\end{lemma}

\begin{proof}
First we bound the probability of successfully performing the subtraction steps given a fixed maximum number of attempts.
We simplify our analysis by assuming that the simulation uses exactly $3r$ subtractions, corresponding to the worst-case scenario in which $2r$ inversions are used.  We want to find the probability that a randomly chosen sequence of subtractions contains at least $r$ successes, correctly implementing the multi-product formula.  The probability that a sequence is unsuccessful is exponentially small in $r$ because for $\delta=1/2$, the mean number of failures is $\mu = 3r/2$, which is substantially smaller than our tolerance of $2r$ failures. By the Chernoff bound, the probability of having more than $2r$ failures satisfies
\begin{equation}
\Pr(X>2r)\le e^{-\mu((1+\alpha)\log(1+\alpha)-\alpha)}<e^{-r/13},
\end{equation}
where $1+\alpha = 2r/\mu = 4/3$.

If we attempt the subtraction steps in our protocol $3r$ times and fail $2r$ times, then $5r$ subtractions and approximate inversions must be performed, because every failure requires an approximate inversion. We bound the resulting error using \lem{MPFSuzError}, \lem{ekinv}, and the subadditivity of errors.  These lemmas apply because the requirement $\frac{4}{3}mk(5/3)^{k-1}ht/r\le\log(2)$ is implied by our choice of $r$ and the assumption $\epsilon\le mhtk^{-4k}$.  By this argument, the simulation error satisfies
\begin{equation}
\norm{\tilde M_{k,k}(t/r)^r-U(t)} \le 5r(2m(5/3)^{k-1}ht/r)^{4k+1}\label{eq:lem11:errorbound}
\end{equation}
with probability at least $1-e^{-r/13}$, where $\tilde M_{k,k}(t/r)$ denotes the operation performed by our non-deterministic algorithm for the multi-product formula $M_{k,k}$.
The assumption $\epsilon\le m h t k^{-4k}$ and the value of $r$ from~\eq{lem11:rbound} imply that $\norm{\tilde M_{k,k}(t/r)^r-U(t)} \le \epsilon$ as required.
\end{proof}

\lem{r} assumes an error tolerance of at most $mhtk^{-4k}$, which may appear to be very small.  However, our ultimate simulation scheme has $k\in O(\sqrt{\log (mht/\epsilon})$, so in fact the error tolerance is modest.

Now we are ready to prove a key lemma that provides bounds on the number of exponentials used by the simulation in the realistic scenario where both addition and subtraction errors may occur.  Our main result, \thm{mainresult}, follows as a simple consequence.

\begin{lemma}\label{lem:nexp}
Let $M_{k,k}$ be a multi-product formula for $H=\sum_{j=1}^m H_j$ as in \defn{MPF}, with $k\ge 1$.  Let $\tilde{M}_{k,k}$ be the implementation of $M_{k,k}$ described above.  Let $\epsilon$ be a desired error tolerance and let $\beta$ be a desired upper bound on the failure probability of the algorithm.  Then there is a simulation of $U(t) = e^{-iHt}$ that has error at most $\epsilon$ with probability at least $1-\beta$ using
\begin{equation}
N_{\exp}\le 1000m5^{k-1}k^{9/4}e^{(1+\log(\eta)/2)k}r \label{eq:thmmain:nexp}
\end{equation}
exponentials of the form $e^{-iH_j t}$, where
\begin{enumerate}
\item $\gamma=\frac{1}{k}\log\lceil \exp([ 1+\log(\eta)/2+\frac{1}{2k}\log(2(2k)^{5/2})]k)\rceil$,
\item $\tilde\epsilon= \min(1,\epsilon, \beta, m h t k^{-4k})$,
\item $r= \left\lceil\max \left\{\frac{(4m(5/3)^{k-1}h t)^{1+1/4k}}{(\tilde\epsilon/5)^{1/4k}},13\log (2/\beta) \right\}\right\rceil$.
\end{enumerate}
\end{lemma}

\begin{proof}
We approximate $U(t)$ as a product of $r$ multi-product formulas, each implemented using the operation $\tilde M_{k,k}$.
We can suppose that the sequence of $r$ multi-product formulas is implemented using at most $5r$ subtraction and inversion steps according to~\lem{r}.  As $k$ unitary operations are combined in each step, we must implement a total of $5r k$ unitary operations.  \defn{MPF} implies that each of these unitaries is composed of at most $e^{\gamma (k+1)}$ Suzuki integrators $U_q$.  Each $U_q$ is a product of $2m5^{k-1}$ exponentials of elements from $\{H_j\}$.  Thus, if the algorithm succeeds after performing this maximum number of subtractions and inversions, we have
\begin{align}
N_{\exp}&\le 10m5^{k-1}ke^{\gamma (k+1)}r\nonumber\\
&\le 10m 5^{k-1}k(2^{11/4}e^2k^{5/4}e^{(1+\log(\eta)/2)k}+1)r\nonumber\\
&< 1000 m 5^{k-1}k^{9/4}e^{(1+\log(\eta)/2)k}r
\end{align}
where we have used $\gamma \le 2$ and $e^{\gamma k}\le 2(2^{7/4}k^{5/4}e^{k(1+\log(\eta)/2)})$.
Equation \eq{thmmain:nexp} then follows by substituting the value of $r$ assumed by the lemma, which guarantees that the simulation error is less than $\epsilon$ when the simulation is successful because it exceeds the value of $r$ from \lem{r}.  We choose $r$ to be larger because it simplifies our results and guarantees that the probability of an addition error is at most
$\tilde \epsilon/2$.

\lem{r} requires $C_q\le 2$, which we have not explicitly assumed.  Substituting the above value of $\gamma$ into the upper bound for $\Sigma_-$ in~\eq{lem2:sumabsCq} shows that $\sum_{q=1}^{k}|C_q|\le 1$, so $|C_q|\le 1$ for all $q=1,\ldots,k$.  Our multi-product formula satisfies $\sum_q C_q=1$, so our choice of $\gamma$ ensures $C_{k+1}\le 2$.  Thus $C_q\le 2$ for all $q$.

\lem{r} implies that the probability of the simulation failing due to too many subtraction errors is at most $e^{-r/13}$.  However, it does not address the possibility of the algorithm failing due to addition errors.  There are at most $5r$ addition steps, so by \thm{main:nonunitary} and the union bound, the probability of an addition error is at most
\begin{equation}
5r P_+\le \frac{5\Delta^2 kr}{4}.\label{eq:mainthm:Delta1}
\end{equation}
By the definition of $\Delta$ in \thm{main:nonunitary},
\begin{align}
\Delta&= \Norm{\max_{q,q'}\Rem_{2k}\left(S_k(t/k_qr)^{k_q}-S_k(t/k_{q'}r)^{k_{q'}}\right) }.
\end{align}
Using \lem{Rl} (similarly as in \eq{lem8:51}), the triangle inequality, and $\frac{4}{3} mk(5/3)^{k-1} h t/r\le \log(2)$, we have
\begin{align}
\Delta&\le 4\frac{(\frac{4}{3}mk (5/3)^{k-1}  h t/r)^{2k+1}}{(2k+1)!}.\label{eq:mainthm:Delta2}
\end{align}
Substituting the assumed value of $r$ into~\eq{mainthm:Delta2} gives
\begin{equation}
\Delta\le \frac{20k^{2k+1}}{(2k+1)!3^{2k+1}}\left(\frac{\tilde \epsilon }{4m(5/3)^{k-1}ht} \right)^{\frac{1}{2}+\frac{1}{4k}}.\label{eq:mainthm:Delta3}
\end{equation}
Substituting this bound into~\eq{mainthm:Delta1} and using~\eq{factorialbd}, we find that the total probability of an addition error is at most
\begin{align}
\frac{20k^{4k+3}\tilde \epsilon}{((2k+1)!)^23^{4k+2}}\left(\frac{\tilde \epsilon}{4mk(5/3)^{k-1}ht} \right)^{1/4k}&\le \frac{20k\tilde \epsilon}{6\pi (6/e)^{4k+2}e^{-2/13}}\left(\frac{\tilde \epsilon}{4mk(5/3)^{k-1}ht} \right)^{1/4k}\label{eq:mainthm:Delta4}.
\end{align}
Since $\tilde \epsilon\le mhtk^{-4k}$, this implies
\begin{equation}
P_+\le \frac{20\tilde \epsilon}{6\pi (6/e)^{4k+2}e^{-2/13}}\left(\frac{1}{4k(5/3)^{k-1}} \right)^{1/4k} < \frac{\tilde \epsilon}{130}< \frac{\tilde \epsilon}{2}.
\end{equation}
Here the second inequality follows from the first by substituting $k=1$, since the middle expression is a monotonically decreasing function of $k$.

The total probability of success $P_s$ satisfies
\begin{equation}
P_s \ge 1-\frac{\tilde \epsilon}{2}-e^{-r/13}.\label{eq:mainthm:psbound1}
\end{equation}
Using $r\ge 13\log(2/\beta)$ and $\tilde \epsilon\le \beta$, we find $P_s \ge 1-\beta$ as claimed.
\end{proof}

As in~\cite{BACS07}, there is a tradeoff between the exponential improvement in the accuracy of the formula and the exponential growth of $M_{k,k}$ with $k$.  To see this, note that apart from terms that are bounded above by a constant function of $k$, $N_{\exp}$ is the product of two terms:
\begin{align}
N_{\exp}&\in O\left(\left[m^2k^{9/4} e^{(1+\log(\eta)/2 +\log(25/3))k}\right]\left[m h t/\tilde\epsilon \right]^{1/4k} \right)\label{eq:scale:nexpscale}\\
&\in O\left(\left[m^2k^{9/4} e^{2.54k}\right]\left[\frac{m h t}{\min(\epsilon,\beta)} \right]^{1/4k} \right).
\end{align}
Here we have not included the $O(\log(1/\beta))$ term from $r$ because $\log(1/\beta)\in O(1/\beta)^{1/4k}$.
For comparison, the results of~\cite{BACS07} and~\cite{WBHS10} have the following complexities:
\begin{align}
N_{\exp}&\in O\left(\left[m^2e^{3.22 k}\right]\left[m h t/\epsilon \right]^{1/2k} \right),\\
N_{\exp}&\in O\left(\left[m^2k e^{2.13 k}\right]\left[m h t/\epsilon \right]^{1/2k} \right),
\end{align}
respectively.  The tradeoff between accuracy and complexity as a function of $k$ is more favorable in our setting than in either of these approaches.  
Finally, we give a detailed analysis of the tradeoff.

\begin{proofof}{\thm{mainresult}}
Neglecting polynomially large contributions in~\eq{scale:nexpscale}, we see that the dominant part of~\eq{scale:nexpscale} is
\begin{equation}
e^{(1+\log(\eta)/2+\log(25/3))k+\frac{1}{4k}\log(mh t/\tilde\epsilon)}.\label{eq:nexpdominant}
\end{equation}
It is natural to choose $k$ to minimize~\eq{nexpdominant}.  The minimum is achieved by taking $k=k_{\rm opt}$, where
\begin{equation}
k_{\rm opt}=\left\lceil\frac{1}{2}\sqrt{\frac{\log(mh t/\tilde\epsilon)}{1+\log(\eta)/2+\log(25/3)}}~\right\rceil\approx 
0.3142 \sqrt{\log(mh t/\tilde\epsilon)}.
\end{equation}
Using this $k$, we find that
\begin{equation}
N_{\exp}\in O\left(k_{\rm opt}^{9/4}m^2h te^{1.6\sqrt{\log(mh t/\tilde\epsilon)}}\right).\label{eq:scale:MPF}
\end{equation}
We have $\tilde \epsilon = \min(1,\epsilon,\beta,mhtk^{-4k})$, so $\tilde\epsilon$ depends implicitly on $k$.  However, we now show that this term can be neglected in the limit of large $mht/\tilde \epsilon$.  In this limit, we have
\begin{equation}
\frac{k_{\rm opt}^{4k_{\rm opt}}}{mht}\in O\left(\frac{{\log(mht/\tilde \epsilon)}^{0.16\sqrt{\log(mht/\tilde \epsilon)}}}{mht} \right)\subset o(1/\tilde \epsilon).
\end{equation}
The above follows since $\lim_{x\rightarrow \infty} \log(x)^{c\sqrt{\log(x)}}/x=0$ for any $c>0$.
Correspondingly, $mhtk_{\rm opt}^{-4k_{\rm opt}}\in \omega (\tilde\epsilon)$.  We therefore conclude that this term can be neglected asymptotically and that $\tilde\epsilon\in \Omega(\epsilon)$, so we can replace $\tilde \epsilon$ with $\epsilon$  asymptotically.
The result then follows by substituting $k_{\mathrm{opt}}$ into~\eq{scale:MPF} and dropping all poly-logarithmic factors.
\end{proofof}

The parameters $\epsilon$ and $\beta$ can be decreased to improve the simulation fidelity and success probability, respectively.  The cost of such an improvement is relatively low, although it is not poly-logarithmic in $1/\epsilon$ and $1/\beta$.  If the initial state of a simulation can be cheaply prepared and the result of the computation can be easily checked, it may be preferable to use large values of $\epsilon$ and $\beta$ and repeat the simulation an appropriate number of times.  Then the Chernoff bound implies that a logarithmic number of iterations is sufficient to achieve success with high probability.  However, this may not be possible if the simulation is used as a subroutine in a larger algorithm (e.g., as in~\cite{HHL09}).


\section{Conclusions}\label{sec:conclusions}

We have presented a new approach to quantum simulation that implements Hamiltonian dynamics using linear combinations, rather than products, of unitary operators.  The resulting simulation gives better scaling with the simulation error $\epsilon$ than any previously known algorithm and scales more favorably with all parameters than simulation methods based on product formulas.
Aside from the quantitative improvement to simulation accuracy, this work provides a new way to address the errors that occur in quantum algorithms.  Specifically, our work shows that approximation errors can be reduced by coherently averaging the results of different approximations.  It is common to perform such averages in classical numerical analysis, and we hope that the techniques presented here may have applications beyond quantum simulation.

It remains an open problem to further improve the performance of quantum simulation as a function of the error tolerance $\epsilon$.  In particular, we would like to determine whether there is a simulation of $n$-qubit Hamiltonians with complexity $\poly(n, \log \frac{1}{\epsilon})$. Classical simulation algorithms based on multi-product formulas achieve scaling polynomial in $\log\frac{1}{\epsilon}$, but they are necessarily inefficient as a function of the system size. Our algorithms fail to provide such favorable scaling in $\frac{1}{\epsilon}$ due to the sign problem discussed in \sec{mpf}. A possible resolution to this problem could be attained by finding multi-product formulas with only positive coefficients and backward timesteps. Such formulas are not forbidden by Sheng's Theorem~\cite{She89} since that result only applies to multi-product formulas that are restricted to use forward timesteps~\cite{Suz91}. New approximation-building methods might use backward timesteps to give multi-product formulas that are easier to implement with our techniques. Conversely, a proof that Hamiltonian simulation with $\poly(n,\log \frac{1}{\epsilon})$ elementary operations is impossible would also show that Lie--Trotter--Suzuki and multi-product formulas with positive coefficients cannot be constructed with polynomially many exponentials, answering an open question in numerical analysis.

Finally, we have focused on the case of time-independent Hamiltonian evolution. One can consider multi-product formulas that are adapted to handle time-dependent evolution, as discussed in~\cite{CG11}. It is nontrivial to use such formulas to generalize our results to the time-dependent case because of difficulties that arise when the Hamiltonian is not a sufficiently smooth function of time. Further investigation of these issues could lead to developments in numerical analysis as well as quantum computing.


\begin{acknowledgments}
This work was supported in part by MITACS, NSERC, the Ontario Ministry of Research and Innovation, QuantumWorks, and the US ARO/DTO.
\end{acknowledgments}


\appendix
\section{Optimality of the linear combination procedure}\label{app:opt}

The goal of this appendix is to show that no protocol for implementing linear combinations of unitary operations in a large family of such protocols can have failure probability less than
\begin{equation}
 \frac{4\kappa}{(\kappa+1)^2},\label{eq:appendix:a1}
\end{equation}
where $\kappa$ is defined in \thm{main:nonunitary}.
Specifically, we consider protocols of the form shown in \fig{appendix:generalcircuit}.
Our result shows that the protocol of \thm{main:nonunitary} is optimal among such all protocols in the limit of small $\Delta$ (i.e., when the unitary operations being combined are all similar).

\begin{theorem}\label{thm:opt}
Any protocol for implementing $V=\sum_{q=0}^{k} C_q U_q$ using a circuit of the form of \fig{appendix:generalcircuit} must fail with probability at least $4\kappa/(\kappa+1)^2$.
\end{theorem}

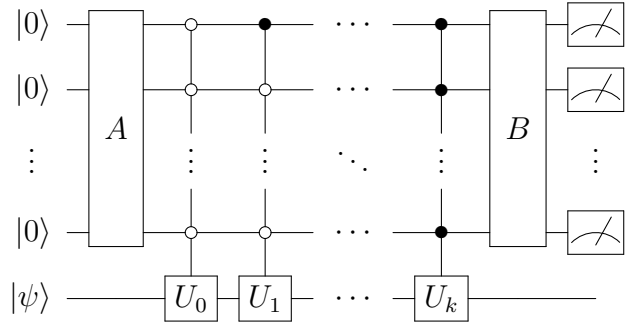
\begin{figure}[t!]
\capstart
\large
\centering
\newcommand{\up}[1]{\push{\raisebox{6pt}{$#1$}}}
\[
\Qcircuit @C=0.7em @R=0.7em {
\lstick{\ket{0}} & \multigate{3}{A} & \ctrlo{1} & \ctrl{1} & \qw & \cdots & & \ctrl{1} & \multigate{3}{B} & \meter \\
\lstick{\ket{0}} & \ghost{A} & \ctrlo{1} & \ctrlo{1} & \qw & \cdots & & \ctrl{1} & \ghost{B} & \meter \\
\lstick{\raisebox{6pt}{\vdots}~\,} & \pureghost{A} & \up{\vdots} & \up{\vdots} & & \up{\ddots} & & \up{\vdots} & \pureghost{B} & \raisebox{6pt}{\vdots} \\
\lstick{\ket{0}} & \ghost{A} & \ctrlo{-1} \qwx[1] & \ctrlo{-1} \qwx[1] & \qw & \cdots & & \ctrl{-1} \qwx[1] & \ghost{B} & \meter \\
\lstick{|\psi\rangle} & \qw & \gate{U_0} & \gate{U_1} & \qw & \cdots & & \gate{U_k} & \qw & \qw
}
\]
\caption{A general circuit for implementing a linear combination of $k+1$ unitary operators using unitary operations $A$ and $B$.  We assume for simplicity that $k+1$ is an integer power of $2$.  This circuit corresponds to preparing the ancilla states in an arbitrary state (specified by $A$) and measuring them in an arbitrary basis (specified by $B$).
\label{fig:appendix:generalcircuit}}
\end{figure}

\begin{proof}
For convenience, we take $k$ to be an integer power of $2$.  We can generalize the subsequent analysis to address the case where $k+1$ is not a power of $2$ by replacing $k+1$ by $k'+1=2^{\lceil \log_2 (k+1) \rceil}$ and taking $C_p=0$ for $p > k$.

Observe that the circuit in \fig{appendix:generalcircuit} acts as follows:
\begin{align}
\ket{0^{ \log_2 k }}\ket{\psi}&\mapsto \sum_{m=0}^{k} A_{m,0} \ket{m}\ket{\psi}\nonumber\\
&\mapsto \sum_m A_{m,0} \ket{m} U_{m} \ket{\psi}\nonumber\\
&\mapsto \sum_{n,m} B_{n,m}A_{m,0} \ket{n}U_{m}\ket{\psi}.\label{eq:appendix:a2}
\end{align}
Furthermore, we can modify $B$ to include a permutation such that desired transformation occurs when the first register is measured to be zero.
(Orthogonality prevents us from having more than one successful outcome, as we show below.)
The implementation is successful if there exists a constant $K>0$ such that
\begin{equation}
B_{0,m}A_{m,0}=K{C_{m}}\label{eq:appendix:a3}
\end{equation}
for all $m$.

Our goal is to maximize the success probability, which is equivalent to maximizing $K$ over all choices of the unitary operations $A$ and $B$ satisfying \eq{appendix:a3}.
We can drop the implicit normalization of the matrix elements of $B$ and $A$ by defining coefficients $b_{0,m}$ and $a_{m,0}$ such that
\begin{align}
\frac{b_{0,m}}{\sqrt{\sum_{j} |b_{0,j}|^2}}&\defeq B_{0,m} \\ \frac{a_{m,0}}{\sqrt{\sum_{j} |a_{j,0}|^2}}&\defeq A_{m,0}.
\end{align}
Using these variables, we have
\begin{equation}
|K C_{m}|
= \frac{|b_{0,m}a_{m,0}|}{\sqrt{(\sum_j |a_{j,0}|^2)(\sum_j |b_{0,j}|^2)}}
\le\frac{|b_{0,m}a_{m,0}|}{\sum_j |a_{j,0}b_{0,j}|}
=\frac{|C_{m}|}{\sum_j |C_j|},
\label{eq:appendix:ambmsum}
\end{equation}
where the bound follows from the Cauchy-Schwarz inequality.
This bound is tight because it can be saturated by taking $a_{m,0}=b_{0,m}=\sqrt{C_m}$.
The probability of successfully implementing the multi-product formula is
\begin{equation}
1-P_- 
= \Big\|{\sum_{j} K C_j U_j\ket{\psi}}\Big\|^2
\le \left( \frac{\sum_{j} C_j }{\sum_{j} |C_j|} \right)^2
=\left(\frac{\Sigma_+ -\Sigma_-}{\Sigma_+ +\Sigma_-}\right)^2,
\label{eq:appendix:pbound0}
\end{equation}
where $\Sigma_+$ and $\Sigma_-$ are defined in \lem{1}.  We have $\kappa\defeq\Sigma_+/\Sigma_-$, so
\begin{equation}
1-P_-\le \left(\frac{\kappa-1}{\kappa+1}\right)^2,\label{eq:appendix:pbound}
\end{equation}
which implies that the failure probability satisfies $P_- \ge 4\kappa/(\kappa+1)^2$ as claimed.

It remains to see why it suffices to consider a single successful measurement outcome.  In principle, we could imagine that many different measurement outcomes lead to a successful implementation of $V$.  Assume that there exists a measurement outcome $v \ne 0^{\log_2 k}$ such that the protocol also gives the same multi-product formula on outcome $v$.
If both outcomes are successful then, up to a constant multiplicative factor, the coefficients of each $U_q$ must be the same.  This occurs if there exists a constant $\Gamma\ne 0$ such that for all $q$,
\begin{equation}\label{eq:appendix:Gamma}
B_{0,q}A_{q,0}=\Gamma B_{v,q}A_{q,0},
\end{equation}
i.e., if $B_{0,q}= \Gamma B_{v,q}$ (note that $A_{q,0}$ must be nonzero provided $C_q \ne 0$).  This is impossible because $B$ is unitary and hence its columns are orthonormal.  Consequently, we cannot obtain the same multi-product formula from different measurement outcomes.
\end{proof}

The above proof implicitly specifies an optimal protocol for implementing linear combinations of unitaries.  However, we do not use this protocol in \thm{main:nonunitary} because it is difficult to perform a correction if the implementation fails.  When the method of \lem{main:2case} fails to implement a difference of nearby unitaries, the desired correction operation is a sum of nearby unitaries, so it can be implemented nearly deterministically.  The correction operation may not have such a form if we use the protocol implicit in the above proof.

One simple generalization of the form of the protocol shown in \fig{appendix:generalcircuit} is to enlarge the ancilla register and allow each unitary in the linear combination to be performed conditioned on a higher-dimensional subspace of the ancilla states.  It can be shown that a certain class of protocols of this form also do not improve the success probability.  Whether protocols of another form could achieve a higher probability of success remains an open question.


\bibliographystyle{apsrev_title}
\bibliography{nonunsim}

\end{document}